\documentclass[12pt]{article}

\usepackage{graphicx}
\usepackage[english]{babel}
\usepackage{amsmath}
\usepackage{hyperref}
\usepackage{amsfonts}
\usepackage{multirow}
\usepackage{amsthm}
\usepackage{wrapfig}
\newcommand{\R}{\mathbb R}
\newcommand{\pr}{\operatorname{Pr}}
\newcommand{\E}{\mathbb{E}}

\newtheorem{theorem}{Theorem}
\newtheorem{definition}{Definition}
\newtheorem{lemma}{Lemma}
\newtheorem{proposition}{Proposition}[section]

\newcommand{\be}{\begin{equation}}
\newcommand{\ee}{\end{equation}}

\begin{document}

\title{Not Always Sparse: Flooding Time in Partially Connected Mobile Ad Hoc Networks}  

\author{Lorenzo Maggi and Francesco De Pellegrini\\
CREATE-NET \\
via alla Cascata, 56/D, 38123 Trento (Italy)\\
E-mail: \{lmaggi,fdepellegrini\}@create-net.org}
\maketitle

\begin{abstract}

In this paper we study mobile ad hoc wireless networks using the notion of evolving connectivity graphs. In such systems, the connectivity changes over time due to the intermittent contacts of mobile terminals. In particular, we are interested in studying the expected flooding time when full connectivity cannot be ensured at each point in time. Even in this case, due to finite contact times durations, connected components may appear in the connectivity graph. Hence, this represents the intermediate case between extreme cases of fully mobile ad hoc networks and fully static ad hoc networks. By using a generalization of edge-Markovian graphs, we extend the existing models based on sparse scenarios to this intermediate case and calculate the expected flooding time. We also propose bounds that have reduced computational complexity. Finally, numerical results validate our models.


\end{abstract}



\section{Introduction}\label{sec:Intro}


Research in ad hoc networks has been focusing on several trade-offs emerging in such systems, and several models able to reproduce their properties have been proposed in literature. In particular, two milestone results in this area captured the properties of fully {\em static} ad hoc networks on the one hand \cite{gupta_capacity}, and the case of fully {\em mobile} ad hoc networks on the other hand \cite{tse_mobility02}. From the constructive proof in~\cite{tse_mobility02} we learn that there exists a class of routing protocols for mobile ad hoc networks which are scalable in the number of nodes, under the hypothesis of \textit{i.i.d.} mobility\footnote{The proof in~\cite{tse_mobility02} was developed for the two hop routing protocol where relays can only receive messages from a source and deliver to the destination.}. 
This fact motivated a large effort in the study of the so called {\em store-carry-and-forward} routing protocols. Such protocols are employed for a specific class of mobile ad hoc networks, namely mobile Delay Tolerant Networks (DTNs). DTNs are networks where connectivity is intermittent and the duration of end-to-end paths is not sufficient to deliver a message from source to destination using traditional routing techniques. Instead, store-carry-and-forward is  used in such systems, meaning that message copies stored at local nodes' memory can be delivered to peer mobiles whenever they enter radio range~\cite{ZNKT}\cite{CBR}\cite{groenevelt2005message}\cite{spy_ton08}. To this respect, mobility compensates for the lack of connectivity: for instance, the so called epidemic routing performs flooding by releasing a copy of the message to any node without the message in radio range. 

Furthermore, in the highly-dynamic scenario depicted above, a customary assumption is that nodes cannot acquire knowledge of the network topology. As a consequence, the process of data exchange in mobile DTNs is typically represented by means of a point process where events are contact times, i.e., the time instants when mobile nodes are in radio range and can exchange message copies \cite{spy_ton08,ABD}. 

However, the effect of mobility is not always such to fully disconnect the network at each point in time. Indeed, full \textit{i.i.d.} mobility captures well the extreme case of networks where the delay for routing a message is dominated by the effect of intermeeting times, i.e., times between successive meeting of two mobile nodes.\footnote{Our definition of intermeeting times refers to the residual time from the end of a contact until a novel contact occurs.} In practice, despite connectivity is evolving and intermittent due to mobility, when contact times have finite duration then connected components may exist.  Thus, within those ``connectivity islands'', the delay for message diffusion needs not to be dominated by the duration of intermeeting times. 
 
In this work, we are interested in the performance figures of the {\em flooding time} in mobile ad hoc  networks in such intermediate regime. 
The flooding time is defined as the first time instant in which all nodes are informed, given that a message is generated by one source node at time $t=0$.

More precisely, we derive performance figures based on a specific model for evolving graphs, which is a generalization of the class of 
edge-Markovian graphs~\cite{ClementiPODC2008}. %
The interest of this topic relates to the fundamental trade-off between mobility and connectivity~\cite{WangTON2013}. In fact, in the case of a static network the classical key concept is that of a connectivity graph. In the simplest case, one may resort to the protocol model proposed in \cite{gupta_capacity}. In turn, it is well known that the flooding time in static networks is dictated by the diameter of the connectivity graph.

However, in a mobile network, the connectivity graph is an {\em evolving} one~\cite{DePe2007}\cite{ferreira_wiopt04}, i.e., links may appear and disappear repeatedly. Hence, the notion of diameter of a graph does not apply to this context.  Rather, as suggested in~\cite{ClementiPODC2008}, the flooding time appears to be the right metric to encode connectivity properties of such networks. In the rest of the paper, we will characterize this metric as a function of the number of nodes of the network and of the mobility parameters of terminals.

The paper is organized as follows. Sect.~\ref{sec:related} describes relevant background on evolving graphs and remarks the novelties introduced in our framework, compared to existing works. In Sect.~\ref{sec:rgraphmodel} we describe our network model, whereas in Sect.~\ref{sec:Poisspoint} we provide an easy expression for the flooding time in the limit case when intermeeting times are much larger than contact durations. In Sect.~\ref{sec:FloodTmain} we develop our main result in the general case. Also, in  Sect.˜\ref{sec:lowboundF} we introduce approximations for the flooding time which require lower computational complexity. Finally, we deliver numerical results in Sect.˜\ref{sec:numres} and a concluding section ends the paper. The main notation used in the rest of the paper is reported in Tab.~\ref{tab:sym}.

\begin{table}[t]\caption{List of symbols used throughout the paper}\label{tab:sym}
\small 
\centering
{\begin{tabular}{|c| p{.8\textwidth}|}
\hline
$\lambda^{-1}$ & : expected contact duration \\ 
$\mu^{-1}$ & : expected intermeeting time \\ 
$p$ & : probability that an edge is in state ON at time 0\\ 
$N$ & : number of nodes \\ 
$F(1,N\!-\!1)$ &  : (also $F(1)$) exact flooding time; expected time for a message generated by 1 node to reach the remaining $N\!-\!1$ nodes\\ 
$F_0(1,N\!-\!1)$ & : (also $F_0(1)$) flooding time in the case of infinitesimal contact duration ($p=0,\mu\uparrow \infty$) \\
$\overline{F}(1,N\!-\!1)$ &  : (also $\overline{F}(1)$) upper bound for $F(1,N\!-\!1)$ \\ 
$\underline{F}(1,N\!-\!1)$ & : (also $\underline{F}(1)$) lower bound for $F(1,N\!-\!1)$ \\
$S_i^{(a)}(N)$ & : event that $i$ nodes are informed, but only $a$ of them are possibly active, i.e., connected with other nodes \\
$W_i^{(c)}(N)$ & : event that $i$ informed nodes are connected to other $c$ nodes \\
$F^{(a)}(i,N-i)$ & : flooding time when $i$ nodes are already informed, but only $a$ of them are possibly active \\
$\mathcal{C}_N,\overline{\mathcal{C}}_N,\underline{\mathcal{C}}_N$ & : computational complexity for $F(1,N\!-\!1)$, $\overline{F}(1,N\!-\!1)$, and $\underline{F}(1,N\!-\!1)$ respectively \\
$G_t$ & : connectivity graph at time $t$\\
\hline
\end{tabular}}
\end{table}


\section{Related works}\label{sec:related}


Models for communications based on evolving graphs have been proposed in literature since the seminal work \cite{ferreira_network}. In \cite{AvinFC}, the authors provide an asymptotic analysis of the cover time of  a graph where edges can be modified at each time step by an adversarial. Covering is shown to be feasible in a time which is polynomial in the size of the graph, despite examples when such time is exponential exist for standard random walkers. In this paper we are interested in studying the flooding time: a similar analysis to the one provided here appears in \cite{ClementiTON2013}. The authors there provide upper and lower bounds for the flooding time in edge-Markovian dynamic graphs, under the assumption that the stationary node distribution at every time step is almost uniform and the transmission radius $r$ is over the connectivity threshold. The model proposed uses discretization in both time and space. In \cite{ClementiPODC2008} the authors prove that the flooding time is $O\Big (\frac{\log N}{\log(1+Np)} \Big )$ and $\Omega\Big (\frac{\log N}{ Np} \Big )$: those results are in line with our results for the continuous-time case. 
 
Flooding time is also a well investigated subject in graph theory. Authors of~\cite{VanMieghemFTRG} extend classical percolation results for the diameter of random graphs \cite{Bollobas}. The relevant result is that the flooding time distribution can be expressed in closed form for the class of Erd\"os Renyi random graphs where link delays are exponentially distributed. Along the same research line, the study in \cite{Armini2013} derives analogous results for regular graphs. Results of \cite{VanMieghemFTRG} resemble the form that we obtain here. However, in our case, the link presence is governed by an on-off process and the topology is not static. In the context of network capacity analysis, indeed, the beneficial effect of mobility has been analysed in the original work \cite{tse_mobility02} and in several follow-ups, including the case with a population of heterogeneous devices characterized by different intermeeting intensities \cite{Garetto2009}. 

The beneficial effect of mobility onto the capacity of ad hoc networks appeared in \cite{tse_mobility02} under the assumption of \textit{i.i.d.} mobility occurs in the unitary disc. Capacity results were extended  to the case of clustered networks in \cite{WangTON2013}: indeed, the presence of a clustered hierarchical structure, is showed to lower the minimum asymptotic degree required to maintain asymptotic connected graphs. The reduction per node obeys to a factor $N^{1-d}$, where $d$ is the cluster size.

Furthermore, several papers extended the analysis to the case of DTNs, where all nodes are disconnected with high probability. In this context, one may model connectivity using evolving graphs and derive delay bounds~\cite{DePe2007}. The authors  of~\cite{Whitbeck2011} identify edge-Markovian dynamic graphs as a promising model to capture the main features of store-carry-and-forward routing protocols.\\

{\noindent \em Novel contributions:} Customary models used in ad hoc networks consider either a fixed topology, i.e., the case of a static ad hoc network, or a fully disconnected network, i.e., the DTN case. The analysis proposed in this paper is able to explore the intermediate case when connectivity is still intermittent as assumed in DTNs, but, whenever a connected component is present, messages can span several hops with very small delay.  Moreover, if compared to the analysis in \cite{ClementiPODC2008}, the continuous time model we employ has the key advantage that the closed forms and bounds proposed in this paper can be specialized to span different conditions of the sub-critical regime, depending on the ratio between the contact time duration and the intermeeting times.


\section{Network model} \label{sec:rgraphmodel}


We consider a delay tolerant network (DTN) with a set $\mathcal{N}$ of $N$ mobile nodes. Any pair of nodes can exchange a message when in radio range, i.e., whenever their distance is smaller than a certain threshold $r>0$. Since terminals roam within a region of size larger than $r$, the connectivity between any two terminals is intermittent.\\
We model this situation by a connectivity graph where nodes represent the terminals and edges evolve over time. At each time $t\ge 0$, the undirected edge between any two nodes can be either in state ON, i.e., the two terminals can communicate with each other, or in state OFF, i.e., they are off-range. The ON--OFF process governing the presence of edge $i$ is stochastic and follows an alternating renewal process. Once an edge is ON, it remains active for a random time $M_1^{(i)}$; then it switches to the OFF state, where it remains for a time $L_1^{(i)}$. It then goes ON for a time $M_2^{(i)}$, and so on and so forth. We assume that the random variables $A^{(i)}=\{M_j^{(i)},L_j^{(i)}\}_{j\ge 1}$ are \textit{i.i.d.}. The intermeeting time $L_j^{(i)}$, i.e., the time between the end of a contact and the next one, is an exponential random variable with mean $\lambda^{-1}>0$. Conversely, we do \emph{not} make further assumptions on the duration of a contact $M_j^{(i)}$: it is a generic random variable with mean $\mu^{-1}\ge 0$.\\
We further suppose that the states of the edges evolve in an independent fashion and according to the same distribution, i.e. $\{A^{(i)}\}_{i=1}^N$ are \textit{i.i.d.}. This assumption allows our problem of computing the flooding time to be analytically tractable.\\
We note that no assumption on the independence between $L_j^{(i)}$ and $M_j^{(i)}$ is required. Hence, our model is still valid if the duration of the contact between any two terminals somehow influences the following intermeeting time.

Let $P(t)$ be the probability that an edge is ON at time $t$. From classic results in renewal theory  (e.g. \cite{ross1996stochastic}, Thm 3.4.4),
\be \label{eq:pdef}
\lim_{t\rightarrow\infty} P(t) = \ \frac{E[M^{(i)}_j]}{E[M^{(i)}_j]+E[L^{(i)}_j]} = \ \frac{\mu^{-1}}{\mu^{-1}+\lambda^{-1}} := \, p.
\ee
We point out that our graph model is a generalization of a continuous time edge-Markovian graph \cite{ClementiPODC2008}, in which the contact duration $M^{(i)}_j$ between any two terminals is also an exponential random variable, independent of $L^{(i)}_j$. In that case, the edge process follows a continuous-time Markov chain with transition rate matrix $Q=\Big (\begin{smallmatrix}
-\lambda & \lambda \\
\mu & -\mu
\end{smallmatrix} \Big )$.

In this paper we are interested in studying the {\em expected flooding time} $F(1,N\!-\!1):=\E[f(1,N\!-\!1)]$, defined as the time employed by a message generated by a tagged node to reach all the remaining $N-1$ nodes. The source copies the message to all the nodes it encounters over time. In turn, all the nodes carrying the message copy it to all the nodes that happen to fall within their communication range. We say that a node is \emph{informed} at time $t$ if, at time $t$, it has a copy of the message. The described relay protocol is called \emph{unrestricted multi-copy protocol} in \cite{groenevelt2005message} or \emph{epidemic routing} \cite{ZNKT}.

In order to develop our analysis, we consider the {\em intermeeting dominated} case, i.e., the transmission time is negligible on the scale of the ON-OFF process governing the link presence: hence, any node \emph{instantaneously} copies the carried message to all the nodes it is connected to.

More specifically, we provide the exact expression of $F(1,N\!-\!1)$. Moreover, so as to facilitate the online estimation of the flooding time, we provide two different upper bounds and one lower bound for $F(1,N\!-\!1)$, which can be computed with lower complexity.

In order to compute the flooding time $F(1,N\!-\!1)$, we assume that the system is observed in steady state\footnote{This is justified if we assume that the network, at time 0, has not been observed for a time long enough (see Eq. \ref{eq:pdef})} at time 0. Thus, if we let $E_t$ be the set of undirected edges in state ON at time $t\ge 0$ and $G_t$ the corresponding connectivity graph, then $G_0\!=\!(E_0,N)$ is an Erd\"os-R\'enyi graph with parameter $p$. 

A notation remark: for compactness' sake, we will drop the dependence of our variables on $N$ whenever possible. For instance, the flooding time $F(1,N\!-\!1)$  will be $F(1)$, and same simplification holds for its lower and upper bounds $\underline{F}$ and $\overline{F}$, respectively. 


\section{Sparse regime: flooding time with point-like contacts} \label{sec:Poisspoint}


Before tackling the computation of the flooding time in the general case with $\lambda^{-1}>0$ and $\mu^{-1}\ge 0$, it is useful to develop a simpler and more restrictive case in which the intermeeting process between any two terminals is a point process, i.e., the contacts' average duration is null (or $\mu\uparrow \infty$)\footnote{We will precise later in Thm.~\ref{thm:scaling} the notion of the sparse regime case as the limit for the general case.}.

In the general model introduced later, more than one link may be active at the same time with positive probability: therefore, in the intermeeting dominated case, the message originated by the source node can be spread to more than one node simultaneously.

Conversely, in the point-like contact case developed in this section, the event that two or more links are in state ON at the same time occurs with null probability. In other words, the set of edges $E_t$ is almost always empty for $t\ge 0$. Moreover, if we suppose that $E_{t}$ contains one edge for some $t\ge 0$, then $E_{t}$ contains another edge with null probability. It is then apparent that in this section we study the {\em sparse regime} of DTNs: a message created by the source can only be transmitted via separate successive hops between pair of terminals.\\
Let us call $F_0(1,N\!-\!1)$ the flooding time in the point-like case. Whenever its dependence on $N$ is clear, we will call it $F_0(1)$. Clearly, $F_0(1)$ constitutes an upper bound for the flooding in the general case for finite $\mu$, i.e.,  
\[
F_0(1,N-1)\ge F(1,N-1), \quad \forall\, \lambda, \, p, \, N.
\]
In fact, the probability of having connected components in an Erd\"os-R\'enyi graph increases with the probability $p$ that an edge is in state ON. In Sect. \ref{sec:2ndupperFlood} we provide a second upper bound for $F(1)$, which is tighter than $F_0(1)$ for large $N$.

To proceed further, let $F_0(i)$ the flooding time with point-like contacts under the hypothesis that $i$ nodes are informed, i.e., they have the message and can forward it to other nodes. Therefore, we can write
\begin{align*}
F_0(N) = & \, 0, \quad F_0(N-1) = \, \frac{1}{\lambda(N-1)}  \\
F_0(i) = & \, \frac{1}{\lambda i(N-i)} + F_0(i+1), \quad 1\le i\le N-2.
\end{align*}
Using the recursive expression above, it follows that
\begin{align*}
F_0(1) = & \, \sum_{i=1}^{N-1} \frac{1}{\lambda i(N-i)} = \, \frac{2}{\lambda N} \sum_{i=1}^{N-1} \frac{1}{i} = \,  \frac{2}{\lambda N} \,H_{N-1},
\end{align*}
where $H_n$ is the $n$-th harmonic number. By the classic relation
\[
\int_1^N \frac{1}{x} \,dx \le H_n \le 1 + \int_2^N \frac{1}{x-1} \,dx,
\]
we easily conclude: $2 \,\frac{\ln N}{\lambda N} \le F_0(1) \le 2\, \frac{1+\ln(N-1)}{\lambda N}$.

Therefore, $F_0(1,N\!-\!1)\in \Theta(N^{-1}\ln N)$ and, since $F_0(1)$ is an upper bound for the flooding time $F(1)$ when $\mu^{-1}>0$, the asymptotic bound for $F(1)$ writes
\[
F(1,N\!-\!1)\in \, O(N^{-1}\ln N).
\]


\section{Flooding time} \label{sec:FloodTmain}


After studying in Sect. \ref{sec:Poisspoint} the flooding time $F_0(1)$ in the specific case of point-like contact ($\mu^{-1}\downarrow 0$), we now investigate the flooding time $F(1)$ in the more general non-sparse case with contact average duration $\mu^{-1}\ge 0$. Hence, $F_0(1)$ equals $F(1)$ calculated at $p=0$.

We recall here two technical assumptions made in Sect. \ref{sec:rgraphmodel}, namely instantaneous transmission during contacts and exponential intermeeting time with mean $\lambda^{-1}$.

Further modelling is required with respect to the point-like contact case of Sect.~\ref{sec:Poisspoint}. In fact, at any time $t$, in the connectivity graph $G_t$ there exist strongly connected components with more than one edge with positive probability. Hence, if at least one node $n$ belonging to a connected component is informed, then the message can be spread instantaneously to all the nodes connected with $n$. The model introduced in this paper is meant precisely to capture such topological feature characterizing mobile ad hoc networks, in order to go beyond the sparse regime.\\
In particular, in this case, we need to distinguish among informed nodes. At each time $t$, in fact, there are two classes of nodes: {\em active} and {\em not active}.  A node is {\em active} if it \emph{possibly} has some link in state ON (each with probability $p$): we hence define the auxiliary notation $F^{(a)}(i)$, which represents the expected flooding time when $i$ nodes are informed and $0\leq a\leq i$ of them are active.

Moreover, 
due to the symmetry of the problem, $F(1)$ does not depend on the identity of the source. 
The exact expression of the flooding time $F(1)$ is presented in iterative form in the following Theorem, whose proof is in the Appendix.

\begin{theorem} \label{theo:floodTexact}
The flooding time $F(1)$ can be expressed as
\begin{align}
F(1) = \, & (1-p)^{N-1}\left( \frac{1}{\lambda(N-1)} + F^{(1)}(2) \right) + \notag \\
& \sum_{c=1}^{N-2}  {N-1\choose c} p^c(1-p)^{N-1-c} \,F^{(c)}(1+c), \label{eq:F1}
\end{align}
where, for $1\le i \le N-2, \ 1\le a\le i-1$,
\begin{small}\begin{align}
F^{(a)}(i) = & \, (1-p)^{a(N-i)}\left( \frac{1}{\lambda i(N-i)} + F^{(1)}(i+1)\right) + \notag\\
& +\sum_{c=1}^{N-i-1} {N-i\choose c} \Big[1-(1-p)^a\Big]^c(1-p)^{a(N-i-c)} \,F^{(c)}(i+c) \label{eq:Fai}
\end{align}\end{small}
and, for $1\le a \le N-2$: $\displaystyle F^{(a)}(N-1) = (1-p)^a\frac{1}{\lambda(N-1)}$. \hfill $\Box$
\end{theorem}

\subsection{Flooding time for $p$ small} \label{sec:floodTpsmall}

In this section we wish to find the relation between the flooding time with point-like contacts, $F_0(1)$, and the general flooding time $F(1)$ for small values of the stationary probability $p$. More specifically, we wish to find an expression of the kind $F(1)=F_0(1)-\chi p+o(p)$, where $\chi$ is some positive constant. We then obtain the following result.
\begin{lemma} \label{lem:floodTpsmall}
The flooding time $F(1,N\!-\!1)$ can be written for values of $p$ sufficiently close to 0 as
\[
F(1,N\!-\!1) = \, F_0(1) - \lambda^{-1} H_{N-1}\, p + o(p),
\]
where $H_n=\sum_{i=1}^n 1/i$ is the $n$-th harmonic number. \hfill $\Box$
\end{lemma}
The proof of Lemma \ref{lem:floodTpsmall} is deferred to the Appendix. We conclude that, for $p$ small, the flooding time decreases approximately linearly in $p$ with coefficient $\lambda^{-1} H_{N-1}$, which behaves like $\log(N)+\gamma$ for large values of $N$, where $\gamma$ is the Euler-Mascheroni constant. 


\subsection{Computational complexity} \label{sec:compcomplfloo}

After providing the exact expression of the flooding time $F(1)$ in Thm. \ref{theo:floodTexact}, in this section we focus on issues related to its computation. First we provide the matricial form of the linear system of equations in (\ref{eq:F1}) and (\ref{eq:Fai}), then we show the complexity of the computation of $F(1)$, in terms of number of required additions and multiplications.
From expressions (\ref{eq:F1}) and (\ref{eq:Fai}), we notice that $(N-2)(N-1)/2$ auxiliary variables need to be utilized in order to compute the flooding time $F(1)$, namely $F^{(a)}(i)$, for $i=2,\dots,N-1$ and $a=1,\dots,i-1$. Hence, the equations in (\ref{eq:F1}) and (\ref{eq:Fai}) can be rewritten as $\mathbf{T}\mathbf{F} = \mathbf{d}$, where $\mathbf{F}$ is the column vector:
\[
\mathbf{F} = \,\left[F^{(1:N-2)}(N-1),F^{(1:N-3)}(N-2),\dots,F^{(1)}(2),F(1)\right]^T
\]
and $F^{(a:a+k)}$ is defined as $[F^{(a)},F^{(a+1)},\dots,F^{(a+k)}]$. $\mathbf{T}$ is a square, lower triangular matrix of dimension $(N-2)(N-1)/2+1$. In order to compute the elements of $\mathbf{T}$ and $\mathbf{d}$, let us define the index mapping:
\begin{align*}
& \Psi(i,a) = \, \frac{(N-1)(N-2)-i(i-1)}{2}+a, \\
& \mathrm{with} \ 2\le i\le N-1, \ 1\le a\le i-1 \ \cup \ (i=1,a=1)
\end{align*}
Then,\\[-5mm]
\begin{align*}
T_{i,i} = & \, 1, \quad 1\le i\le (N-2)(N-1)/2+1\\
T_{\Psi(i,a),\Psi(i+1,1)} = & \, -(1-p)^{a(N-i)} + \\
 & \ \ -(N-i)\left[1-(1-p)^a\right](1-p)^{a(N-i-1)} \\
T_{\Psi(i,a),\Psi(i+c,c)} = & \, -{N-i\choose c} \Big[1-(1-p)^a\Big]^c(1-p)^{a(N-i-c)}, \\
& \quad 2\le c\le N-i-1
\end{align*}
and $\displaystyle d_{\Psi(i,a)} = \, (1-p)^{a(N-i)}\frac{1}{\lambda i(N-i)}$. For example, when $N=5$, the linear system has the following structure:
\[
\begin{small}
\begin{bmatrix}
1 & 0 & 0 & 0 & 0 & 0 & 0 \\
0 & 1 & 0 & 0 & 0 & 0 & 0 \\
0 & 0 & 1 & 0 & 0 & 0 & 0 \\
\bullet & 0 & 0 & 1 & 0 & 0 & 0 \\
\bullet & 0 & 0 & 0 & 1 & 0 & 0 \\
0 & \bullet & 0 & \bullet & 0 & 1 & 0 \\
0 & 0 & \bullet & 0 & \bullet & \bullet & 1 \\
\end{bmatrix} \,
\begin{bmatrix}
F^{(1)}(4) \\
F^{(2)}(4) \\
F^{(3)}(4) \\
F^{(1)}(3) \\
F^{(2)}(3) \\
F^{(1)}(2) \\
F(1)
\end{bmatrix} = \, \mathbf{d}
\end{small}
\]
where the symbol $\bullet$ stands for the generic nonnull element.

Now we are ready to compute the complexity of the flooding time: the proof requires the enumeration of the operations 
performed in order to solve the linear system $\mathbf{T}\mathbf{F} = \mathbf{d}$ and it is omitted for the sake of space.

\begin{proposition} \label{prop:complFexact}
The number $\mathcal{C}_N$ of operations (i.e.  multiplications or additions) required to compute $F(1,N\!-\!1)$ both equal the number of nonnull elements of the matrix $\mathbf{T}$ below the main diagonal, i.e. $\mathcal{C}_N = (N^3-6N^2+17N-18)/6$. \hfill $\Box$
\end{proposition}


\subsection{Scaling law for the flooding time in the sub-critical regime.}


With our model we can bound the expected flooding time as a function of the network size $N$. To this respect we referred to the point-like bound $F_0(1)$ as the limit case when  $\mu$ diverges. However, it is interesting to compute the scaling law for the network when such bound is actually the limit case for large $N$'s. In fact, our analysis addresses the case when the network is in the sub-critical regime, i.e., with a terminology of random graphs analysis, the graph is disconnected almost surely.

In particular, with respect to the connectivity properties of the system, the dynamic graph parameters $\lambda=\lambda(N)$ and $\mu=\mu(N)$ may depend on the underlying mobility model in non trivial fashion. The question in turn is under which scaling regime the point-like bound is guaranteed to be still the limit of a scaling law which is consistent with the sub-critical regime (when the graph is connected indeed we expect the bound to be zero). In other words: as $N$ diverges, how should the contact time duration and the intermeeting time scale in order to approach the point-wise limit while still in the sub-critical regime for every $N$? In practice, consistent scaling laws 
are ensured by the following result.

\begin{theorem}\label{thm:scaling}
Let $\frac \mu\lambda = \Omega\big( \frac N{\log N}\big)$, then $F(1,N\!-\!1)=O(N^{-1} \log N)$. \hfill $\Box$
\end{theorem}
\begin{proof}
Consider the class of graphs $G(N,p)$ where $p=p(N)=\frac{ \lambda}{\mu +  \lambda}$: according to the classic result 
\cite{Bollobas}[Thm 7.3 pp. 164], we can fix a $c \in \R$, and if  $p=\{\log N + c + o(1)\}/N$, then $\lim_{N\rightarrow \infty}P(G_p \;\mbox{is connected} ) \rightarrow e^{e^{-c}}$. Here, under our assumptions, indeed,  there exists $N_0$ such that $\frac \lambda\mu > b \frac{N}{\log N}$, where $b>0$, which implies for $N\geq N_0$
\[
p(N)=\frac{\lambda(N)}{\lambda(N)+\mu(N)}=\frac{1}{1+\frac \mu\lambda}< \frac{\lambda}\mu < \frac 1{b} \frac{\log N}N
\]
so that $p(N)=o\big ( \frac{\log N} N \big )$. Since the point-like case is an upper bound, the statement follows.
\end{proof}

Thm.~\ref{thm:scaling} explains to which extent the point process bound is the limit for our model: this is the case when the intermeeting time grows much faster than the contact duration. Under this law, the network approaches the conditions of a sparse network which is typically used in literature for DTN models \cite{ABD,groenevelt2005message}.


\section{Approximating the flooding time with lower complexity}\label{sec:lowboundF}


In Thm.~\ref{theo:floodTexact} we provided the exact expression of the flooding time $F(1)$. It is the solution of a linear system of equations, which can be solved via iterative substitution. It can be calculated at a cost of $\sim N^3/6$ operations (see Prop. \ref{prop:complFexact}). If the flooding time needs to be computed online, there may arise the need of reducing its computational complexity. Therefore, in this section we aim at approximating $F(1)$ with lower complexity, by providing a lower bound $\underline{F}(1)$ and an upper bound $\overline{F}(1)$ for $F(1)$ in Sects. \ref{sec:lowboundF} and \ref{sec:2ndupperFlood}, respectively.

\subsection{A lower bound for the flooding time} \label{sec:lowbsub}

In this section we intend to provide a lower bound $\underline{F}(1)$ for the flooding time $F(1)$, whose computation complexity is bounded asymptotically by the complexity of $F(1)$ (see Prop. \ref{prop:complFexact}). To this aim, we need to simplify the expression in (\ref{eq:F1}) by formulating some convenient approximations. Before going deeper into the analysis, we still need to define two concepts.
\begin{definition}
$S_i^{(a)}(N)$ is defined as the event that, among $i$ \emph{i}nformed terminals, only $a\le i$ of them are possibly \emph{a}ctive, i.e. connected with some of the $N-i$ uninformed nodes. In details, let $\mathcal{I}\subset\mathcal{N}$ be the set of informed terminals ($|\mathcal{I}|=i$) and let $\mathcal{A}\subseteq\mathcal{I}$ ($|\mathcal{A}|=a$). Then, $i)$ all the $a(N-i)$ edges between $\mathcal{A}$ and $\mathcal{N}\setminus\mathcal{I}$ and $ii)$ all the edges within $\mathcal{N}\setminus\mathcal{I}$ are independently in state ON with probability $p$. Also, $iii)$ all edges between $\mathcal{I}\setminus\mathcal{A}$ and $\mathcal{N}\setminus\mathcal{I}$ are OFF with probability $1$.
\end{definition}
\begin{figure}[h]
\centering
\includegraphics[width=.6\textwidth]{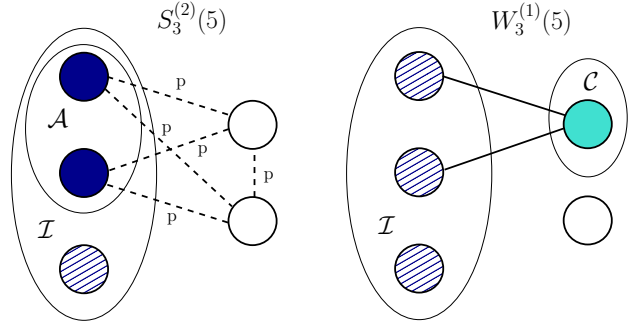}
\caption{Illustrations of an instance of the event $S_3^{(2)}(5)$ (on the left side) and of $W_3^{(1)}(5)$ (on the right side). The dashed edges are ON with probability $p$.}
\label{fig:Sia}
\end{figure}
\begin{definition}
$W_i^{(c)}(N)$ is defined as the event that $c$ nodes are \emph{c}onnected to the $i$ informed nodes. More formally, let $\mathcal{I}\subset\mathcal{N}$ be the set of informed terminals ($|\mathcal{I}|=i$) and let $\mathcal{C}\subseteq\mathcal{N}\setminus\mathcal{I}$ ($|\mathcal{C}|=c$). Then, for any terminal $n\in \mathcal{C}$ there exists a link in state ON between $n$ and an informed node $n'\in \mathcal{I}$. Moreover, all the edges between $\mathcal{I}$ and $\mathcal{N}\setminus(\mathcal{I}\cup\mathcal{C})$ are in state OFF with probability 1. Hence, the message is forwarded instantaneously to all the nodes in $\mathcal{C}$.
\end{definition}

It is straightforward to see that the following relation holds:
\be \label{eq:correct}
\left(W_i^{(c)}(N),S_i^{(a)}(N)\right) = \ S_{i+c}^{(c)}(N).
\ee
In other words, the message is transmitted at time $t^+$ to the $c$ connected nodes, which are the only ones which can possibly forward the message to some of the remaining $N-i-c$ uninformed nodes at time $t^+$. In fact, the event $(W_i^{(c)},S_i^{(a)})$ rules out the possibility that the $a$ nodes possibly active at time $t$ are capable to retransmit the message at time $t^+$.

Instead, in order to compute a lower bound for $F(1)$ we \emph{assume} that all the $i+c$ informed nodes are still able to spread the message to the other nodes with probability $p$. More formally, we replace the correct equation on the left-hand side of the following expression:
\be \label{eq:apprLB}
\left(W_i^{(c)}(N),S_i^{(a)}(N)\right) = \ S_{i+c}^{(c)}(N) \ \stackrel{\mathrm{replace}}{\longrightarrow} \ S_{i+c}^{(i+c)}(N).
\ee
with the right-hand side one (compare with Eq. \ref{eq:correct}). In other words, by (\ref{eq:apprLB}) we overestimate the activity of the edges, by assigning a probability of being active $p$ also to those edges that are known being in state OFF. Under this approximation, the flooding process speeds up, and we can lower bound the correct expression of $F(1)$ in (\ref{eq:F1}) with $\underline{F}(1)$, that we define in the following.
\begin{proposition} \label{prop:lowboundF1}
The flooding time under the approximation in (\ref{eq:apprLB}) is $\underline{F}(1)$, where
\begin{align*}
\underline{F}(1) = & \ (1-p)^{N-1}\left[ \frac{1}{\lambda(N-1)} +  \underline{F}(2)\right] + \\ & \ + \sum_{i=1}^{N-2} {N-1\choose i} p^{i}(1-p)^{N-1-i} \,\underline{F}(i+1),
\end{align*}
and where $\displaystyle \underline{F}(i) = \ (1-p)^{i(N-i)}\left[ \frac{1}{\lambda i(N-i)} +  \underline{F}(i+1)\right] +$
\[
+ \sum_{c=1}^{N-i-1} {N-i\choose c} \left[1-(1-p)^i\right]^c(1-p)^{i(N-i-c)} \,\underline{F}(i+c). \label{eq:lowerT}
\]
for $2\le i\le N-2$, and $\underline{F}(N-1)=(1-p)^{N-1}/(\lambda(N-1))$. Moreover, $\underline{F}(1)$ is a lower bound for the flooding time $F(1)$. \hfill $\Box$
\end{proposition}

We notice that the expression of $\underline{F}(1)$ can be computed iteratively, starting from $\underline{F}(N-1)$, and then utilizing the expressions of $\underline{F}(N-1),\dots,\underline{F}(i+1)$ to compute $\underline{F}(i)$, for $i=N-2,\dots,1$.\\

\subsubsection{Lower bound for $p$ small} \label{sec:lowboundpsmall}

Under Assumption (\ref{eq:apprLB}), we actually consider the OFF edges as having a probability $p$ of being ON. Intuitively, if the stationary probability $p$ tends to 0, $\underline{F}(1)$ should approximate the actual flooding time $F(1)$ with increasing accuracy. We now prove this formally and we also provide an expression of the lower bound $\underline{F}(1)$ for $p$ small, in a similar fashion as we did for the exact flooding time $F(1)$ in Sect.~\ref{sec:floodTpsmall}.

\begin{lemma} \label{lem:Flowpsmall}
The lower bound $\underline{F}(1)$ of the flooding time can be written for values of $p$ sufficiently close to 0 as
\[
\underline{F}(1) = \, F_0(1) -\lambda^{-1}(N-1)\,p+ o(p). \quad \Box
\] 
\end{lemma}
\vspace{.2cm}

Lemma \ref{lem:Flowpsmall} confirms the intuition that when $p$ tends to 0, then $\underline{F}(1)\rightarrow F_0(1)$. Moreover, in a right neighbourhood of $p=0$, $\underline{F}(1)$ decreases linearly in $p$, and proportionally to the number of nodes $N$. The reader should compare this result with Lemma \ref{lem:floodTpsmall}, claiming that the exact flooding time $F(1)$ decreases as $\lambda^{-1} H_{N-1} p$ for $p$ small, where $H_N$ diverges as $\log(N)$.\\

\subsubsection{Computational complexity}

In order to compute $\underline{F}(1)$ we need to solve an $(N\!-\!1)$-by-$(N\!-\!1)$ linear system of equations of the form $\mathbf{A} \, \underline{\mathbf{F}} = \, \mathbf{c}$ (see proof of Lemma \ref{lem:Flowpsmall}), where the unknowns are $\underline{F}(i)$ for $i=1,\dots,N\!-\!1$. The number of operations required to solve the linear system equals the number of elements of the matrix $\mathbf{A}$ above the main diagonal. The following result follows. We leave its proof to the reader.

\begin{proposition} \label{prop:compllow}
The number of operations required to compute the flooding time lower bound $\underline{F}(1,N,1)$ is $\underline{\mathcal{C}}_N={{N-1}\choose{2}}$. \hfill $\Box$
\end{proposition}

By comparing the computational complexity required to compute the exact flooding time $F(1)$ and its lower bound $\underline{F}(1)$ (see Prop. \ref{prop:complFexact} and \ref{prop:compllow}, respectively), we see that $\underline{\mathcal{C}}_N < \mathcal{C}_N$ for $N\ge 5$, and $\underline{\mathcal{C}}_N =\mathcal{C}_N$ for $N\le 4$. Moreover,
\[
\underline{\mathcal{C}}_N \sim \, \frac{N^2}{2}, \quad \mathcal{C}_N \sim \, \frac{N^3}{6}.
\]
Hence, approximating the exact expression of the flooding time $F(1)$ via the lower bound $\underline{F}(1)$ becomes more and more convenient when the number of terminals $N$ increases.

\subsection{An upper bound for the flooding time} \label{sec:2ndupperFlood}

After providing a lower bound for the flooding time $F(1)$ with low complexity, we now devise an upper bound, called $\overline{F}(1,N\!-\!1)$. Although its computational complexity equals the one of the exact flooding time, its expression is recursive in the number of terminals $N$. Hence, the computation of $\overline{F}(1,N\!-\!1)$ already yields $\overline{F}(1,K-1)$ for all $K<N$. As such, the complexity required to calculate the first $N$ values of the bound is $O(N^3)$, whereas $O(N^4)$ operations are needed for the exact value of the flooding time.\\
We observed in Sect. \ref{sec:Poisspoint} that the exact flooding time with point-like contacts, $F_0(1,N\!-\!1)$, already provides an upper bound for $F(1,N\!-\!1)$. We compare these two bounds in Sect. \ref{sec:numres}.

Let us describe the assumption which allows us to derive the expression of the upper bound $\overline{F}$. Let $\mathcal{I}$ be the set of informed nodes at time $t$. As soon as a connection between $\mathcal{I}$ and some $\mathcal{C}\subseteq\mathcal{N}$ is established at time $t'\ge t$, the message is copied instantaneously to $\mathcal{C}$. Then, we \emph{assume} that all the terminals in $\mathcal{I}$ deactivate and no longer participate in the flooding process. This corresponds  to progressively remove informed nodes from the original system  and to consider the flooding process on the remaining subgraph.

More specifically, in order to derive the upper bound we perform the following substitution:
\be \label{eq:apprUB}
\left(W_i^{(c)}(N),S_i^{(a)}(N)\right) = \ S_{i+c}^{(c)}(N) \ \stackrel{\mathrm{replace}}{\longrightarrow} \ S_c^{(c)}(N-i).
\ee
Clearly, the spreading process under the approximation (\ref{eq:apprUB}) is slower than in the exact case because the number of nodes capable of retransmitting the message decreases over time.\\ 
Thus, the flooding time calculated under assumption (\ref{eq:apprUB}) provides an upper bound for $F(1)$, namely $\overline{F}(1)$.
\begin{proposition} \label{prop:Floodupp}
The flooding time under the approximation in (\ref{eq:apprUB}) is $\overline{F}(1)$, where \vspace{-.2cm}
\begin{align*}
\overline{F}(1) = & \ (1-p)^{N-1}\left[ \frac{1}{\lambda(N-1)} +  \overline{F}(1,N-2)\right] + \\ & \ + \sum_{c=1}^{N-2} {N-1\choose c} p^c(1-p)^{N-1-c} \,\overline{F}(c,N-c-1),
\end{align*}
where, for $2\le i \le n-2$,
\begin{align*}
\overline{F}(i,n-i) = & \ (1-p)^{i(n-i)}\left[ \frac{1}{\lambda i(n-i)} +  \overline{F}(1,n-i-1)\right] + \\ \sum_{c=1}^{n-i-1} {n-i\choose c} & \Big[1-(1-p)^{i}\Big]^c(1-p)^{i(n-i-c)} \,\overline{F}(c,n-i-c), \\
\overline{F}(n-1,1) = & \ (1-p)^{n-1}\frac{1}{\lambda(n-1)}.
\end{align*}
Moreover, $\overline{F}(1)$ is an upper bound for the exact flooding time $F(1)$. \hfill $\Box$
\end{proposition}
We defer the proof of Proposition \ref{prop:Floodupp} to the Appendix.\\

\subsubsection{Computational complexity}

We now investigate the computational complexity $\overline{\mathcal{C}}_N$ of the upper bound $\overline{F}(1,N\!-\!1)$, defined as the number of operations (i.e., additions or multiplications) required to compute $\overline{F}(1,N\!-\!1)$.

\begin{proposition} \label{prop:complUB}
Let $\overline{\mathcal{C}}_N$ be the number of operations to compute $\overline{F}(1,N\!-\!1)$. Then, $\overline{\mathcal{C}}_N=\mathcal{C}_N$ for all integers $N$. \hfill $\Box$
\end{proposition}
\vspace{2pt}

Interestingly, the complexity of $\overline{F}(1,N\!-\!1)$ equals the complexity of the exact flooding time $F(1,N\!-\!1)$.
\begin{figure}
\centering
\includegraphics[width=.6\textwidth]{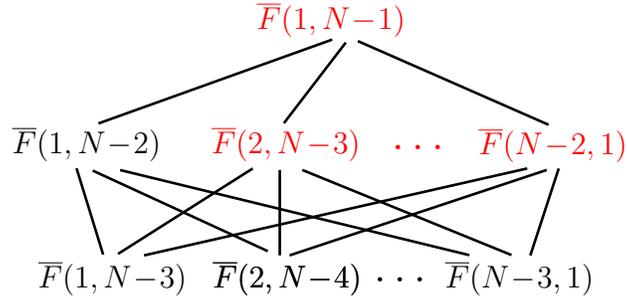}
\caption{The quantities in black are available when $\overline{F}(1,N\!-\!2)$ has already been computed. The quantities in red need to be calculated in order to derive $\overline{F}(1,N\!-\!1)$.}
\label{fig:Upp_compl}
\end{figure}
As anticipated, approximating $F$ as $\overline{F}$ is  computationally profitable when we need to evaluate $\overline{F}(1,N\!-\!1)$ for several values of $N$. Once $\overline{F}(1,N\!-\!2)$ has already been computed indeed, in order to compute the $\overline{F}(1,N\!-\!1)$ we only need to perform an incremental number of operations equal to (compare with Eq.~\ref{eq:mispiace} and Fig.~\ref{fig:Upp_compl})
\[
(N-2) + (N-4)+(N-5)+\dots+1= \, \frac{(N-1)(N-2)}{2}-N+3.
\]

\section{Numerical results}\label{sec:numres}



We begin our numerical investigations by displaying
\begin{figure}
\centering 
\includegraphics[width=.6\textwidth]{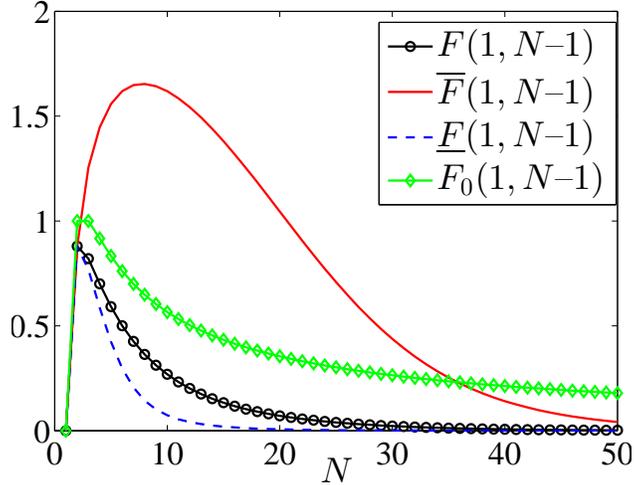}
\caption{Exact flooding time $F(1,N\!-\!1)$ ($\lambda=1$, $p=0.12$), compared with its lower bound $\underline{F}(1,N\!-\!1)$, its upper bound $\overline{F}(1,N\!-\!1)$, and the flooding time in sparse regime $F_0(1,N\!-\!1)$.}
\label{fig:flood4lines}
\end{figure}
in Fig. \ref{fig:flood4lines} the exact flooding time $F(1)$ when the size of the network $N$ varies. We also illustrate the bounds $\underline{F}(1)$, $\overline{F}(1)$, and the flooding time in sparse regime $F_0(1)$. In the next sections we go deeper in the analysis by showing numerically the relations among these quantities for different values of $p$ and $N$.

\subsection{Not always sparse: Comparison with punctual case}

The main objective of this paper is to account for the existence, at each instant, of connected components in the connectivity graph of intermittently connected mobile ad hoc networks. We have shown that the sparse regime usually studied in DTNs is a special instance of our model in the limit when intermeeting times dominate contact times ($p=0$). When this is not the case, i.e., $p>0$, one may be tempted to use still the same model, which is certainly appealing for the simplicity of the expression of flooding time $F_0$ in that regime (see Section \ref{sec:Poisspoint}). Nevertheless, we now show numerically that, assuming that the regime is non-sparse ($p>0$), the error committed by approximating the flooding time in the sparse regime ($p=0$) may be tremendously high even for reasonable values of the stationary probability $p$. In Figure \ref{fig:approxerror}(a) we illustrate the ratio between the flooding time in the sparse regime, $F_0(1,N\!-\!1)$, and the flooding time in the non-sparse regime, $F(1,N\!-\!1)$, for different values of $p$ and $N$. As expected, the approximation error is negligible for $p$ small. Nevertheless, it becomes indeed unacceptable for $p>0.1$ even for networks with a relatively small number of nodes ($N\le 50$).

\subsection{Bounds tightness}

We now report some numerical results on the tightness of the lower and upper bounds $\underline{F}(1)$ and $\overline{F}(1)$, proposed in Sections \ref{sec:lowbsub} and \ref{sec:2ndupperFlood} respectively. Similarly to Fig. \ref{fig:approxerror}(a), in Fig. \ref{fig:approxerror}(b) we illustrate the ratio between the upper bound $\overline{F}(1)$ and the exact expression of the flooding time $F(1)$, for different values of the stationary probability $p$ and the number of terminals $N$. Fig. \ref{fig:approxerror}(c) shows the same analysis for the lower bound $\overline{F}(1)$. We notice that the upper bound $\overline{F}(1)$ is a good approximation of $F(1)$ for values sufficiently large of $p$, i.e. $p>0.3$. On the other hand, Fig. \ref{fig:approxerror}(c) confirms our analytic results in Sect. \ref{sec:lowboundpsmall}: the lower bound $\underline{F}(1)$ well approximates the flooding time $F(1)$ for small values of $p$ (also, see Lemma \ref{lem:Flowpsmall}).

\subsection{Comparison between the two upper bounds}

In this paper we provided the expression of two different upper bounds for the flooding time $F(1,N\!-\!1)$, namely $F_0(1,N\!-\!1)$ (Sect. \ref{sec:Poisspoint}) and $\overline{F}(1,N\!-\!1)$ (Sect. \ref{sec:2ndupperFlood}). The former is actually the exact expression of the flooding time $F(1,N\!-\!1)$ in the sparse regime. Fig. \ref{fig:upp_bound_compare} illustrates the comparison between these two bounds.  Numerical experiments showed that the bound $\overline{F}(1,N\!-\!1)$ is tighter than $F_0(1,N\!-\!1)$ under two different regimes, namely $i)$ for all $N$, if $p$ is sufficiently large, or $ii)$ for $N$ large enough, if $p$ is sufficiently small. In other words, if we fix $\lambda$ and $p$, then there exists $\widehat{N}$ such that $\overline{F}(1,N\!-\!1) < F_0(1,N\!-\!1), \ \forall\, N>\widehat{N}$. Hence, the case $\widehat{N}=0$ describes regime $ii)$.

\begin{figure}
\centering 
\includegraphics[width=.6\textwidth]{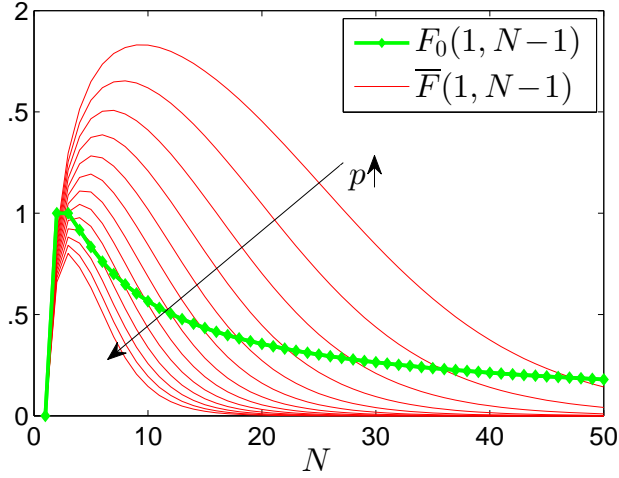}
\caption{Upper bounds $F_0$ and $\overline{F}$ compared for $\lambda=1$. $\overline{F}(1,N\!-\!1)<F_0(1,N\!-\!1)$ either $i)$ for all $N$, if $p>0.3$ or $ii)$ for $N$ sufficiently large for $p<0.3$.}
\label{fig:upp_bound_compare}
\end{figure} 


\begin{figure*}[t!] 
\centering 
\includegraphics[width=\textwidth]{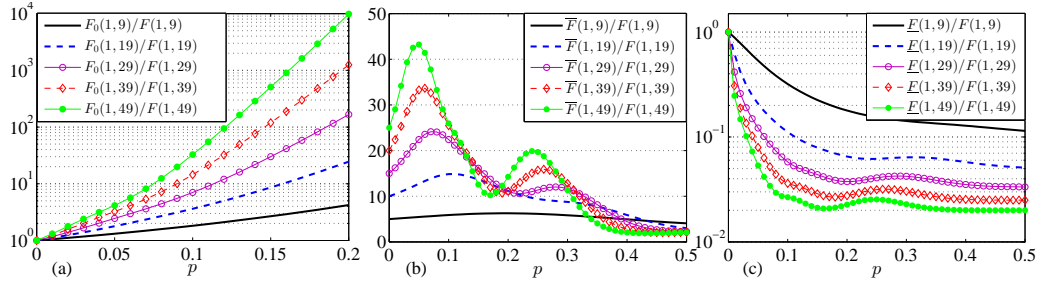}
\caption{In (a) we show the ratio $F_0(1)/F(1)$ for different values of $p$ and $N=10,\dots,50$: the approximation error becomes larger than $100\%$ for $p>0.1$ as soon as $N\geq 10$. At $p=0.1$, when $N \geq 50$, $F_0(1) > 10 F(1)$. Similarly, (b) and (c) illustrate the ratios $\overline{F}(1)/F(1)$ and $\underline{F}(1)/F(1)$, respectively. We see that $\underline{F}(1)$ well approximates $F(1)$ for small values of $p$, while $\overline{F}(1)$ approaches the exact value $F(1)$ for $p$ sufficiently large, i.e., $p>0.3$.}
\label{fig:approxerror}
\end{figure*}


\section{Conclusions}\label{sec:concl}


In this paper we focused on the flooding time for mobile ad hoc networks. In this context, the notion of diameter for static networks has little meaning. In our case, in fact, the network is disconnected almost surely at each point in time, so that the diameter is infinite. Yet, mobility of terminals overcomes instantaneous lack of connectivity and flooding time is still well defined. We applied a generalization of continuous time Markov-edge evolving graphs to this context. The advantage of such models is that they allow to encode the impact of finite contact durations into the flooding time calculation, where the event of possibly large connected components cannot be neglected. In such cases, the classic sparse models used in DTNs fail to capture such events thus providing conservative estimations. We provided the exact expression of the flooding time $F(1)$ in the non-sparse regime. Our continuous time model encompasses the sparse model as limit case. We computed a lower bound $\underline{F}(1)$ and an upper bound $\overline{F}(1)$, both having a reduced computational complexity. We studied analytically the behaviour of $F(1),\underline{F}(1),\overline{F}(1)$ for small values of the stationary probability $p$ and we investigated numerically the bounds tightness with respect to the exact value $F(1)$. We further showed the error committed by approximating $F(1)$ with the equivalent expression in the sparse regime, $F_0(1)$.
 
There are two key directions that we have not explored in this work. First, the heterogeneity of the link ON-OFF processes has been neglected for tractability's sake. An extension of the model in that direction would include the case when terminals have different mobility patterns, and this would permit tighter estimations. Also, in real mobility models, it is possible to measure a certain degree of correlation among intermeeting events. We plan to include such correlation effects into our model in our future research work.


\appendix
\section*{Proof of Theorem \ref{theo:floodTexact}}
\begin{proof}
For the sake of notation simplicity, we drop the dependence of the events $S_i^{(a)}$ and $W_i^{(c)}$ on $N$. We observe that the joint event $(W_i^{(c)},S_i^{(a)})$, with $a\le i$, is equivalent to the event $S_{i+c}^{(c)}$. In fact, when the set $\mathcal{I}$ of informed nodes transmit instantaneously the message to $\mathcal{C}$, then the number of informed nodes becomes $\mathcal{I}\cup\mathcal{C}$, but only the nodes in $\mathcal{C}$ can possibly copy the message to $\mathcal{N}\setminus(\mathcal{I}\cup\mathcal{C})$ instantaneously. By the total probability rule, we write $F(1) := \E [f(1)|S_1^{(1)}]$ as
\begin{align}
F(1) =& \, \sum_{c=0}^{N-1} \pr\left(W_1^{(c)}|S_1^{(1)}\right) \, \E\left[f(1)|W_1^{(c)},S_1^{(1)}\right] \notag\\
= & \, \sum_{c=0}^{N-2} \pr\left(W_1^{(c)}|S_1^{(1)}\right) \, \E\left[f(1)|S_{c+1}^{(c)}\right], \label{eq:F1n1}
\end{align}
\begin{align}
\E\left[f(1)|S_i^{(a)}\right] = & \, \sum_{c=0}^{N-i}\pr\left( W_i^{(c)}|S_i^{(a)} \right) \E\left[f(1)|W_i^{(c)},S_i^{(a)}\right] \notag\\
= & \, \sum_{c=0}^{N-i-1}\pr\left( W_i^{(c)}|S_i^{(a)} \right) \E\left[f(1)|S_{i+c}^{(c)}\right], \label{eq:Fhnh}
\end{align} 
with $a<i$. We remark that $\E[f(1)|S_N]=0$ because the transmission is instantaneous. Due to this, we could truncate the sums in (\ref{eq:F1n1}) and (\ref{eq:Fhnh}) up to $c=N-i-1$. Now, we find an explicit expression for (\ref{eq:F1n1}). When $i$ nodes are informed and all the $i(N-i)$ edges between $\mathcal{I}$ and $\mathcal{N}\setminus\mathcal{I}$ are OFF, an exponential time with mean $1/(\lambda i(N-i))$ needs to be waited before the flooding process resumes. At that instant, there are $i+1$ informed nodes but only one of them, i.e. the newly informed one, can possibly instantaneously copy the message to the uninformed nodes. Therefore, we can writ
\[
\E\left[f(1)|S_i^{(0)}\right] = \,  \frac{1}{\lambda i(N-i)} + \E\left[f(1)|S_{i+1}^{(1)}\right]. 
\]
Now we derive an explicit expression for the probability terms in (\ref{eq:F1n1}) and (\ref{eq:Fhnh}). The event that $c$ nodes (set $\mathcal{C}$) are connected to $i$ informed nodes (set $\mathcal{I}$), of which only $a\le i$ are possibly active (set $\mathcal{A}$), is the intersection of the events $i)$ for each $n\in\mathcal{C}$, there exists $n'\in\mathcal{A}$ such that the edge $(n,n')$ is ON, and $ii)$ all the edges between $\mathcal{A}$ and $\mathcal{N}\setminus(\mathcal{I}\cup\mathcal{C})$ are in state OFF. Thus, we can write 
\be \label{eq:remind}
\pr\left( W_i^{(c)}|S_i^{(a)} \right) = {N-i\choose c} \Big[1-(1-p)^{a}\Big]^c(1-p)^{a(N-i-c)}.
\ee
Let us define $F^{(a)}(i):=\E[f(1)|S_{i}^{(a)}]$. We can interpret $F(1)$ as $F^{(1)}(1)$. Then, the thesis easily follows by inspection.
\end{proof}

\section*{Proof of Proposition \ref{prop:lowboundF1}}

\begin{proof}
By the substitution in (\ref{eq:apprLB}), we modify (\ref{eq:F1n1}) as\\ $F'(1)= \, \sum_{c=0}^{N-2} \pr\left(W_1^{(c)}|S_1^{(1)}\right) \, \E\left[f(1)|S_{1+c}^{(1+c)}\right]$. \\
By defining $F'(i):=\E [f(1)|S_{i}^{(i)}]$, under assumption (\ref{eq:apprLB}) we modify (\ref{eq:Fhnh}) as $F'(i) = \, \sum_{c=0}^{N-i-1}\pr ( W_i^{(c)}|S_i^{(i)} ) F'(i+c).$ Let us define $\underline{F}(i):=F'(i)$, $1\le i\le N-1$. We compute $\pr (W_i^{(c)}|S_i^{(i)})$ as in (\ref{eq:remind}). The thesis follows by inspection.
\end{proof}

\section*{Proof of Lemma \ref{lem:Flowpsmall}}

\begin{proof}
Let us rewrite (\ref{eq:lowerT}) as $\mathbf{A} \, \underline{\mathbf{F}} = \, \mathbf{c}$, where $\underline{\mathbf{F}}= [\underline{F}(1) \ \dots \ \underline{F}(N-1)]^F$, $\mathbf{A}$ is an upper triangular $(N\!-\!1)$-by-$(N\!-\!1)$ matrix with diagonal elements $A_{i,i} = \, 1$ and 
\begin{align*}
A_{i,i+1} = & \, -(1-p)^{i(N-i)} -(N-i)\left[1-(1-p)^i\right] \times \\
& \ \ \ \times (1-p)^{i(N-i-1)}= \, -1+o(p) \\
A_{i,i+c} = & \, -{N-i\choose c} \left[1-(1-p)^i\right]^c (1-p)^{i(N-i-c)} = \, o(p),\\
c_i(p) = & \,  \frac{(1-p)^{i(N-i)}}{\lambda i(N-i)}= \, \frac{1-i(N-i)p}{\lambda i(N-i)} + o(p)
\end{align*} 
It is not difficult to show that $\underline{F}(N-i,i)=\sum_{i=N-i}^{N-1}(1-i(N-i)p)/(\lambda i(N-i))+o(p)$, and in particular 
\begin{align}
\underline{F}(1) = & \, \sum_{i=1}^{N-1} \frac{(1-p)^{i(N-i)}}{i(N-i)}-\lambda^{-1}(N-1)\,p+ o(p) \notag\\
= & \, F_0(1) - \lambda^{-1}(N-1)\,p+ o(p) \notag
\end{align}
\end{proof}

\section*{Proof of Proposition \ref{prop:Floodupp}}

\begin{proof}
By performing the substitution in (\ref{eq:apprUB}), we modify (\ref{eq:F1n1}) as $F''(1,N\!-\!1)=$\\
$\sum_{c=0}^{N-2} \pr\left(W_1^{(c)}(N)|S_1^{(1)}(N)\right) \, \E\left[f(1,N-2)|S_{c}^{(c)}(N-1)\right]$.\\
We define $F''(i,n\!-\!i)\!:=\!\E [f(1)|S_{i}^{(i)}(n)]$, $n\!<\!N$. Under assumption (\ref{eq:apprUB}) we modify (\ref{eq:Fhnh}) as $F''(i,n-i) = \, \sum_{c=0}^{N-i-1}\pr\left( W_i^{(c)}(n)|S_i^{(i)}(n) \right) F''(c,n\!-\!i\!-\!c).$ Let us define $\overline{F}(i,n-i):=F''(i,n-i)$. The thesis follows by inspection.
\end{proof}

\section*{Proof of Lemma \ref{lem:floodTpsmall}}

\begin{proof}
We consider the matricial expression $\mathbf{T}\mathbf{F}=\mathbf{d}$ (see Sect. \ref{sec:compcomplfloo}) and write the elements of $\mathbf{d},\mathbf{T}$ for $p$ small:
\begin{align*}
d_{\Psi(i,a)} = & \, \frac{1}{\lambda i(N-i)}- \frac{a}{\lambda i}p+o(p)\\ := & \,\alpha_{\Psi(i,a)} - \beta_{\Psi(i,a)}p + o(p)\\
T_{\Psi(i,a),\Psi(i+1,1)} = & \, -1+o(p) \\
T_{\Psi(i,a),\Psi(i+c,c)} = & \, o(p), \quad c=2,\dots,N-i-1,
\end{align*}
and $T_{i,i} = 1 $. We can write for $2\le i \le N-2$: 
\begin{align*}
&F^{(1)}(i) = \ \alpha_{\Psi(i,1)} - \beta_{\Psi(i,1)}\,p + F^{(1)}(i+1) + o(p), \\
&F^{(1)}(N-1) =  \ \alpha_{\Psi(N-1,1)} - \beta_{\Psi(N-1,1)}\,p + o(p)\\
&F(1) =  \ \alpha_{\Psi(1,1)} - \beta_{\Psi(1,1)}\,p + F^{(1)}(2) + o(p).
\end{align*}
Hence, we find that 
\begin{align*}
F(1,N\!-\!1) \!=\! \sum_{i=1}^{N-1} d_{\Psi(i,1)} = & \sum_{i=1}^{N-1} \frac{1}{\lambda i(N-i)} - \sum_{i=1}^{N-1} \frac{1}{\lambda i}\,p + o(p)\\
= & F_0(1) - \lambda^{-1} H_{N-1}\, p + o(p)\notag  
\end{align*}
\end{proof}

\section*{Proof of Proposition \ref{prop:complUB}}

\begin{proof}
Let us prove the thesis by induction on $N$. We know that $\overline{\mathcal{C}}_1= \overline{\mathcal{C}}_2 = 0$. Suppose that we have computed $\overline{F}(1,N-2)$ with $\overline{\mathcal{C}}_{N-1}$ operations. Then, we also have at our disposal $\overline{F}(i,N-i-1)$ for $1\le i \le N-3$. In order to compute $\overline{F}(1,N\!-\!1)$, we first need to derive $\overline{F}(i,N-1-i)$ for $i=2,\dots,N-2$ (see Fig. \ref{fig:Upp_compl}). The computation of $\overline{F}(i,N-1-i)$ requires $N-i-2$ operations. Finally, calculating $\overline{F}(1,N\!-\!1)$ involves $N-2$ operations. Hence, for $N\ge 3$,
\begin{align}
\overline{\mathcal{C}}_N = & \, \overline{\mathcal{C}}_{N-1} + (N\!-\!2) + (N\!-\!4)+(N\!-\!5)+\dots+1 \label{eq:mispiace} 
\end{align}
Therefore, $\overline{\mathcal{C}}_N = \sum_{i=1}^{N-4} i(N-3-i) + (N-1)(N-2)/2$, which equals $\mathcal{C}_N$. 

\end{proof}


\begin{thebibliography}{10}

\bibitem{ABD}
E.~Altman, T.~Basar, and F.~De Pellegrini.
\newblock Optimal monotone forwarding policies in delay tolerant mobile ad-hoc
  networks.
\newblock {\em Elsevier Perform. Eval.}, 67(4):299--317, 2010.

\bibitem{Armini2013}
H.~Amini, M.~Draief, and M.~Lelarge.
\newblock {Flooding in Weighted Sparse Random Graphs}.
\newblock {\em SIAM Journal on Discr. Math.}, 27(1):1--26, January 2013.

\bibitem{AvinFC}
C.~Avin, M.~Kouck\'{y}, and Z.~Lotker.
\newblock How to explore a fast-changing world (cover time of a simple random
  walk on evolving graphs).
\newblock In {\em Proc. of ICALP}, pages 121--132, Reykjavik, Iceland, 2008.

\bibitem{Bollobas}
B\'ela Bollob\'as.
\newblock {\em Random Graphs}.
\newblock Cambridge University Press, 2001.

\bibitem{CBR}
A.~Chaintreau, J.-Y.~Le Boudec, and N.~Ristanovic.
\newblock The age of gossip: Spatial mean-field regime.
\newblock In {\em Proc. of ACM Sigmetrics}, June 2009.

\bibitem{ClementiPODC2008}
A.~Clementi, C.~Macci, A.~Monti, F.~Pasquale, and R.~Silvestri.
\newblock Flooding time in edge-markovian dynamic graphs.
\newblock In {\em Proc. of PODC}, pages 213--222, New York, NY, USA, 2008. ACM.

\bibitem{ClementiTON2013}
A.~Clementi, F.~Pasquale, and R.~Silvestri.
\newblock Opportunistic manets: Mobility can make up for low transmission
  power.
\newblock {\em Networking, IEEE/ACM Transactions on}, 21(2):610--620, 2013.

\bibitem{DePe2007}
F.~De~Pellegrini, D.~Miorandi, I.~Carreras, and I.~Chlamtac.
\newblock In {\em Proc. of IEEE INFOCOM}.

\bibitem{ferreira_network}
A.~Ferreira.
\newblock Building a reference combinatorial model for {MANET}s.
\newblock {\em IEEE Network}, 8(5):24--29, Sept./October 2005.

\bibitem{ferreira_wiopt04}
A.~Ferreira and A.~Jarry.
\newblock Complexity of minimum spanning tree in evolving graphs and the
  minimum-energy broadcast routing problem.
\newblock In {\em Proc. of WiOpt'04}, Cambridge, UK, March 24--26, 2004.

\bibitem{Garetto2009}
M.~Garetto, P.~Giaccone, and E.~Leonardi.
\newblock Capacity scaling in ad hoc networks with heterogeneous mobile nodes:
  The super-critical regime.
\newblock {\em Networking, IEEE/ACM Trans. on}, 17(5):1522--1535, 2009.

\bibitem{groenevelt2005message}
R.~Groenevelt, P.~Nain, and G.~Koole.
\newblock The message delay in mobile ad hoc networks.
\newblock {\em Performance Evaluation}, 62(1):210--228, 2005.

\bibitem{tse_mobility02}
M.~Gr\"ossglauser and D.~Tse.
\newblock Mobility increases the capacity of ad hoc wireless networks.
\newblock {\em IEEE/ACM Transactions on Networking}, 10:477--486, Aug. 2002.

\bibitem{gupta_capacity}
P.~Gupta and P.~R. Kumar.
\newblock The capacity of wireless networks.
\newblock {\em IEEE Trans. on Information Theory}, 46(2):388--404, March 2000.

\bibitem{ross1996stochastic}
Sheldom~M Ross.
\newblock {\em Stochastic processes, 2nd ed.}
\newblock Wiley, New York, 1996.

\bibitem{spy_ton08}
T.~Spyropoulos, K.~Psounis, and C.~Raghavendra.
\newblock Efficient routing in intermittently connected mobile networks: the
  multi-copy case.
\newblock {\em ACM/IEEE Trans. on Networking}, 16:77--90, Feb. 2008.

\bibitem{VanMieghemFTRG}
R.~{Van Der Hofstad}, G.~Hooghiemstra, and P.~{Van Mieghem}.
\newblock The flooding time in random graphs.
\newblock {\em Extremes}, 5(2):111--129, 2002.

\bibitem{WangTON2013}
X.~Wang, X.~Lin, Q.~Wang, and W.~Luan.
\newblock Mobility increases the connectivity of wireless networks.
\newblock {\em Networking, IEEE/ACM Trans. on}, 21(2):440--454, 2013.

\bibitem{Whitbeck2011}
J.~Whitbeck, V.~Conan, and M.~Dias~de Amorim.
\newblock Performance of opportunistic epidemic routing on edge-markovian
  dynamic graphs.
\newblock {\em Communications, IEEE Trans. on}, 59(5):1259--1263, 2011.

\bibitem{ZNKT}
X.~Zhang, G.~Neglia, J.~Kurose, and D.~Towsley.
\newblock Performance modeling of epidemic routing.
\newblock {\em Elsevier Computer Networks}, 51:2867--2891, July 2007.

\end{thebibliography}

\end{document}